\documentclass[11pt,onecolumn,twoside]{IEEEtran}
\usepackage{amsmath, amssymb, graphicx, cite, subfigure}

\newtheorem{cor}{Corollary}
\newtheorem{defin}{Definition}
\newtheorem{prop}{Proposition}
\newtheorem{lemma}{Lemma}

\newtheorem{remark}{Remark}
\newtheorem{algorithm}{Algorithm}
\newtheorem{assum}{Assumption}

\def\var{\mathop{\rm var}}
\DeclareMathOperator{\E}{E}
\def\Bern{\mathop{\rm Bern}}

\DeclareMathOperator*{\argmax}{arg\,max}

\title{Channels That Die}%
\author{Lav R. Varshney,  
   	  Sanjoy K.~Mitter,       
        and Vivek K Goyal%
\thanks{This work was supported in part by the NSF Grants CCR-0325774 and CCF-0729069.}%
\thanks{This work appeared in part in the Proceedings of the Forty-Seventh Annual Allerton Conference 
	  on Communication, Control, and Computing \cite{VarshneyMG2009}.}
\thanks{L.~R.~Varshney was with the Department of Electrical Engineering and Computer Science, 
        the Research Laboratory of Electronics, and the Laboratory for Information and Decision Systems,
	  Massachusetts Institute of Technology, Cambridge, MA 02139 USA.  He is now with the
	  IBM Thomas J.\ Watson Research Center, Hawthorne, NY 10532 USA (e-mail: varshney@alum.mit.edu).}%
\thanks{S.~K.\ Mitter is with the Department of Electrical Engineering and Computer Science, the Engineering Systems Division,
	  and the Laboratory for Information and Decision Systems,
	  Massachusetts Institute of Technology, Cambridge, MA 02139 USA (e-mail: mitter@mit.edu).}%
\thanks{V.~K.\ Goyal is with the Department of Electrical Engineering and Computer Science and
        the Research Laboratory of Electronics,
        Massachusetts Institute of Technology,
        Cambridge, MA 02139 USA (e-mail: vgoyal@mit.edu).}%
}

\begin{document}
\maketitle

\begin{abstract}
Given the possibility of communication systems failing
catastrophically, we investigate limits to communicating 
over channels that fail at random times.  These channels are
finite-state semi-Markov channels.  We show that communication with
arbitrarily small probability of error is not possible.  
Making use of results in finite blocklength channel coding,
we determine sequences of blocklengths that optimize transmission
volume communicated at fixed maximum message error probabilities.
We provide a partial ordering of communication channels.
A dynamic programming formulation is used to show the structural result
that channel state feedback does not improve performance.
\end{abstract}

\vspace{2mm}
{\small ``a communication channel{\ldots} might be inoperative because of an amplifier failure, 
a broken or cut telephone wire, \ldots''} \\
\hspace*{\fill} {\small\emph{--- I.~M.~Jacobs \cite{Jacobs1959}}} \\

\section{Introduction}

Physical systems have a tendency to fail at random times \cite{Davis1952}.
This is true whether considering communication systems embedded in sensor networks that may run 
out of energy \cite{DietrichD2009}, synthetic communication systems embedded in
biological cells that may die \cite{CantonLE2008},\footnote{We sidestep teleological 
discussions of natural biology \cite{Pfeifer2006,Varshney2006} by considering 
synthetic biology \cite{Endy2005}.}
communication systems embedded in spacecraft that may enter black holes \cite{Bekenstein2001},
or communication systems embedded in oceans with undersea cables that may be cut \cite{Headrick1991}.
In these scenarios and beyond, failure of the communication system may be modeled as 
communication channel death.

As such, it is of interest to study information-theoretic limits on communicating
over channels that die at random times.  This paper gives results on the fundamental 
limits of what is possible and what is impossible when communicating over channels that die.  
Communication with arbitrarily small probability of error (\emph{Shannon reliability}) is not possible 
for any positive communication volume, however a suitably defined notion of $\eta$-reliability is possible.
Schemes that optimize communication volume for a given level of $\eta$-reliability are developed herein.

The central trade-off in communicating over channels that die is in the lengths
of codeword blocks.  Longer blocks improve communication performance as 
classically known, whereas shorter blocks have a smaller probability of
being prematurely terminated due to channel death.  In several settings, a simple greedy
algorithm for determining the sequence of blocklengths yields a certifiably optimal solution.  
We also develop a dynamic
programming formulation to optimize the ordered integer partition that determines
the sequence of blocklengths.  Besides algorithmic utility, solving the dynamic program demonstrates
the structural result that channel state feedback does not improve performance.

The optimization of codeword blocklengths is reminiscent of 
frame size control in wireless networks \cite{HaraOAOM1996,Modiano1999,LettieriS1998,CiSN2005}, 
however such techniques are used in conjunction with automatic repeat request protocols
and are motivated by amortizing protocol information.
Moreover, the results demonstrate the benefit of adapting to either channel state or 
decision feedback.  Contrarily, we show that adaptation to channel state provides 
no benefit for channels that die.

Limits on channel coding with finite blocklength \cite{Slepian1963,ChangHM1962,MacMullenC1998,Laneman2006,PolyanskiyPV2008,BuckinghamV2008,WiechmanS2008}
are central to our development.  
Indeed, channels that die bring the notion of finite blocklength to the fore and provide
a concrete physical reason to step back from infinity.\footnote{The phrase ``back from infinity'' is borrowed
from J.\ Ziv's 1997 Shannon Lecture.}  Notions of outage in wireless communication \cite{OzarowSW1994,Goldsmith2005} 
and lost letters in postal channels \cite{WolfWZ1970}
are similar to channel death, except that neither outage nor lost letters are permanent conditions.  
Therefore blocklength asymptotics are useful to study those channel models but are not
useful for channels that die.  Recent work that has similar motivations as this paper 
provides the outage capacity of a wireless channel \cite{ZengZC2008}.

The remainder of the paper is organized as follows.  Section~\ref{sec:model}
defines discrete memoryless channels that die and shows that these channels have
zero Shannon capacity.  Section~\ref{sec:systemmodel} states the communication
system model and also fixes our novel performance criteria.  Section~\ref{sec:limcomm}
shows that our notion of Shannon reliability is not achievable, strengthening the result of zero Shannon capacity
and then provides a communication scheme and determines its performance.  Section~\ref{sec:optim}
optimizes performance for several death distributions using either a greedy algorithm or
a dynamic programming algorithm.  Optimization demonstrates
that channel state feedback does not improve performance.  
Section~\ref{sec:ordering} discusses the partial ordering of channels.
Section~\ref{sec:conc} suggests several extensions to this work.

\section{Channel Model}
\label{sec:model}

Consider a channel with finite input alphabet $\mathcal{X}$
and finite output alphabet $\mathcal{Y}$.  It has an \emph{alive} state $s = a$ when it acts like a 
noisy discrete memoryless channel (DMC) and a \emph{dead} state $s = d$ when it erases the input.\footnote{Our results can be extended to
cover cases where the channel acts like other channels \cite{PolyanskiyPV2009,PolyanskiyPV2009b} 
in the alive state.} 
Assume throughout the paper that the DMC from the alive state has zero error capacity \cite{Shannon1956}
equal to zero.\footnote{If the channel is noiseless in the alive state, the problem is 
similar to settings where fountain codes \cite{Sanghavi2007} are used in the point-to-point case and 
growth codes \cite{KamraMFR2006} are used in the network case.}

For example, if the channel acts like a binary symmetric channel (BSC) with crossover
probability $0 < \varepsilon < 1$ in the alive state, with $\mathcal{X} = \{0,1\}$, and 
$\mathcal{Y} = \{0,1,?\}$, then the transmission matrix in the alive state is
\begin{equation}
\label{eq:pa}
p(y|x,s=a) = p_a(y|x) = \begin{bmatrix} 1 - \varepsilon & \varepsilon & 0 \\
				\varepsilon & 1 - \varepsilon & 0 \end{bmatrix} \mbox{,}
\end{equation}
and the transmission matrix in the dead state is
\begin{equation}
\label{eq:pd}
p(y|x,s=d) = p_d(y|x) = \begin{bmatrix} 0 & 0 & 1 \\
				0 & 0 & 1 \end{bmatrix} \mbox{.}
\end{equation}

The channel starts in state $s=a$ and then transitions to $s=d$ at some random time $T$,
where it remains for all time thereafter.  That is, the channel is in state $a$ 
for times $n = 1,2,\ldots,T$ and in state $d$ for times $n = T+1,T+2,\ldots$. 
The death time distribution is denoted $p_T(t)$.  Note that there is always
a finite $t^{\dagger}$ such that $p_T(t^{\dagger}) > 0$.

\subsection{Finite-State Semi-Markov Channel}
Channels that die can be classified as finite-state channels (FSCs) \cite[Sec.~4.6]{Gallager1968}.
\begin{prop}
\label{prop:FSC}
A channel that dies $(\mathcal{X},p_a(y|x),p_d(y|x),p_T(t),\mathcal{Y})$
is a finite-state channel.
\end{prop}
\begin{IEEEproof}
Follows by definition, since the channel has two states.
\end{IEEEproof}

Channels that die have semi-Markovian \cite[Sec.~4.8]{Ross1996}, \cite[Sec.~5.7]{Gallager1996} properties.
\begin{defin}
A semi-Markov process changes state according to a Markov chain
but takes a random amount of time between changes.  More specifically,
it is a stochastic process with states from a discrete alphabet $\mathcal{S}$, such that
whenever it enters state $s$, $s \in \mathcal{S}$:
\begin{itemize}
  \item The next state it will enter is state $r$ with probability that depends only on $s,r \in \mathcal{S}$.
  \item Given that the next state to be entered is state $r$, the time until the transition from $s$ to $r$ 
	occurs has distribution that depends only on $s,r \in \mathcal{S}$.
\end{itemize}
\end{defin}
\begin{defin}
The Markovian sequence of states of a semi-Markov process is 
called the embedded Markov chain of the semi-Markov process.
\end{defin}
\begin{defin}
A semi-Markov process is irreducible if its embedded Markov chain is irreducible.
\end{defin}
\begin{prop}
\label{prop:sM}
A channel that dies $(\mathcal{X},p_a(y|x),p_d(y|x),p_T(t),\mathcal{Y})$ has a channel state sequence that 
is a non-irreducible semi-Markov process.
\end{prop}
\begin{proof}
When in state $a$, the next state is $d$ with probability $1$ and given that
the next state is to be $d$, the time until the transition from $a$ to $d$ has distribution $p_T(t)$.
When in state $d$, the next state is $d$ with probability $1$.
Thus, the channel state sequence is a semi-Markov process.

The semi-Markov state process is not
irreducible because the $a$ state of the embedded Markov chain is transient.
\end{proof}
\noindent Note that when $T$ is a geometric random variable, the channel state
process forms a Markov chain, with transient state $a$ and recurrent, absorbing state
$d$.

There are further special classes of FSCs.
\begin{defin}
An FSC is a \emph{finite-state semi-Markov channel} (FSSMC) if its
state sequence forms a semi-Markov process.
\end{defin}
\begin{defin}
An FSC is a \emph{finite-state Markov channel} (FSMC) if its
state sequence forms a Markov chain.
\end{defin}
\begin{prop}
A channel that dies $(\mathcal{X},p_a(y|x),p_d(y|x),p_T(t),\mathcal{Y})$
is an FSSMC and is an FSMC when $T$ is geometric.
\end{prop}
\begin{IEEEproof}
Follows from Props.~\ref{prop:FSC} and \ref{prop:sM}.
\end{IEEEproof}

FSMCs have been widely studied in the literature \cite{Gallager1968,Gallager1972,TatikondaM2009}, particularly the
panic button/child's toy channel of Gallager \cite[p.~26]{Gallager1972}, \cite[p.~103]{Gallager1968} and the 
Gilbert-Elliott channel and its extensions \cite{MushkinB1989,GoldsmithV1996}.

Contrarily, FSSMCs seem to not have
been specifically studied in information theory. There are a few works \cite{BratenT2002,WangCA2008,WangP2010} that
give semi-Markov channel models for wireless communications systems but do not
provide information-theoretic characterizations.

\subsection{Capacity is Zero}
A channel that dies has Shannon capacity equal to zero.  To show this, first
notice that if the initial state of a channel that dies were not fixed, then it would
be an indecomposable FSC \cite[Sec.~4.6]{Gallager1968}, where the effect of the initial state dies
away.

\begin{prop}
\label{prop:indec}
If the initial state of a channel that dies $(\mathcal{X},p_a(y|x),p_d(y|x),p_T(t),\mathcal{Y})$
is not fixed, then it is an indecomposable FSC.
\end{prop}
\begin{IEEEproof}
The embedded Markov chain for a channel that dies has a unique absorbing state
$d$.
\end{IEEEproof}

Indecomposable FSCs have the property that the upper capacity, defined in \cite[(4.6.6)]{Gallager1968},
and lower capacity, defined in \cite[(4.6.3)]{Gallager1968},
are identical \cite[Thm.~4.6.4]{Gallager1968}.  This can be used to show
that the capacity of a channel that dies is zero.
\begin{prop}
\label{prop:cap_zero}
The Shannon capacity, $C$, of a channel that dies $(\mathcal{X},p_a(y|x),p_d(y|x),p_T(t),\mathcal{Y})$
is zero.
\end{prop}
\begin{IEEEproof}
Although the initial state is $s_1=a$ here, temporarily suppose 
that $s_1$ may be either $a$ or $d$.  Then the channel is 
indecomposable by Prop.~\ref{prop:indec}.

The lower capacity $\underline{C}$ equals the upper capacity $\overline{C}$,  
for indecomposable channels by \cite[Thm.~4.6.4]{Gallager1968}. 
The information rate of a memoryless $p_d(y|x)$ `dead' channel is clearly zero
for any input distribution, so the lower capacity $\underline{C} = 0$.  
Thus the Shannon capacity for a channel that dies with initial alive state is $C = \overline{C} = 0$.
\end{IEEEproof}

\section{Communication System}
\label{sec:systemmodel}
In order to information theoretically characterize a channel that dies, a communication system
that contains the channel is described.

We have an information stream (like i.i.d.\ equiprobable bits), which can be grouped into a sequence of 
$k$ messages, $(W_1,W_2,\ldots,W_k)$.  
Each message $W_i$ is drawn from a message set $\mathcal{W}_i = \{1,2,\ldots,M_i\}$.  
Each message $W_i$ is encoded into a channel input codeword $X_1^{n_i}(W_i)$ and these codewords
$(X_1^{n_1}(W_1),X_1^{n_2}(W_2),\ldots,X_1^{n_k}(W_k))$ are transmitted in sequence over the channel.  A noisy version
of this codeword sequence is received, $Y_1^{n_1+n_2+\cdots+n_k}(W_1,W_2,\ldots,W_k)$.  The receiver then 
guesses the sequence of messages using an appropriate decoding rule $g$,
to produce $(\hat{W}_1,\hat{W}_2,\ldots,\hat{W}_k) = g(Y_1^{n_1+n_2+\cdots+n_k})$.  
The $\hat{W}_i$s are drawn from alphabets $\mathcal{W}_i^{\ominus} = \mathcal{W}_i \cup \ominus$,
where the $\ominus$ message indicates the decoder declaring an erasure.  The receiver
makes an error on message $i$ if $\hat{W}_i \neq W_i$ and $\hat{W}_i \neq \ominus$.

Block coding results are typically expressed with the concern of
sending one message rather than $k$ messages as here.\footnote{Tree 
codes are beyond the scope of this paper, since we desire to communicate messages.  A reformulation of
communicating over channels that die using tree codes \cite[Ch.~10]{Jelinek1968b} with early termination \cite{Forney1974a}
would, however, be interesting.  In fact, communicating over channels that die using convolutional
codes with sequential decoding would be very natural, but would require performance criteria different from the ones 
developed herein.}

System definitions can be formalized as follows.
\begin{defin}
An $(M_i,n_i)$ \emph{individual message code} for a channel that dies $(\mathcal{X},p_a(y|x),p_d(y|x),p_T(t),\mathcal{Y})$ 
consists of:
\begin{enumerate}
  \item An individual message index set $\{1,2,\ldots,M_i\}$, and
  \item An individual message encoding function $f_{i}: \{1,2,\ldots,M_i\} \mapsto \mathcal{X}^{n_i}$.
\end{enumerate}
The individual message index set $\{1,2,\ldots,M_i\}$ is denoted $\mathcal{W}_i$, and the 
set of individual message codewords $\{f_{i}(1),f_{i}(2),\ldots,f_{i}(M_i)\}$ is called the 
\emph{individual message codebook}.
\end{defin}

\begin{defin}
An $(M_i,n_i)_{i=1}^k$ \emph{code} for a channel that dies $(\mathcal{X},p_a(y|x),p_d(y|x),p_T(t),\mathcal{Y})$ 
is a sequence of $k$ individual message codes, $(M_i,n_i)_{i=1}^k$, in the sense of comprising:
\begin{enumerate}
  \item A sequence of individual message index sets $\mathcal{W}_1, \mathcal{W}_2, \ldots, \mathcal{W}_k$,
  \item A sequence of individual message encoding functions $f = (f_1,f_2,\ldots,f_k)$, and 
  \item A decoding function $g: \mathcal{Y}^{\sum_{i=1}^k n_i} \mapsto \mathcal{W}_1^{\ominus} \times \mathcal{W}_2^{\ominus} \times \cdots \times \mathcal{W}_k^{\ominus}$.
\end{enumerate}
\end{defin}

There is no essential loss of generality
by assuming that the decoding function $g$ is decomposed into a sequence of
individual message decoding functions $g = (g_1,g_2,\ldots,g_n)$ where $g_{i}: \mathcal{Y}^{n_i} \mapsto \mathcal{W}_i^{\ominus}$
when individual messages are chosen independently,
due to this independence and the conditional memorylessness of the channel.

To define performance measures, we assume that the decoder operates
on an individual message basis.  That is, when applying the communication system, 
let $\hat{W}_1 = g_1(Y_1^{n_1})$, $\hat{W}_2 = g_2(Y_{n_1 + 1}^{n_1 + n_2})$, and so on.

For the sequel, we make a further assumption on the operation of the decoder. 
\begin{assum}
\label{assum:erase}
If all $n_i$ channel output symbols used by individual message decoder $g_i$
are not $?$, then the range of $g_i$ is $\mathcal{W}_i$.  If any of
the $n_i$ channel output symbols used by individual message decoder $g_i$
are $?$, then $g_i$ maps to $\ominus$.
\end{assum}
This assumption corresponds to the physical properties of a communication system
where the decoder fails catastrophically.  Once the decoder fails, it cannot perform
any decoding operations, and so the $?$ symbols in the channel model of system failure
must be ignored.

\subsection{Performance Measures}

We formally write the notion of error for the communication system as follows.
\begin{defin}
For all $1 \le w \le M_i$, let
\[
\lambda_{w}(i) = \Pr[\hat{W}_i \neq w|W_i = w, \hat{W}_i \neq \ominus]
\]
be the conditional message probability of error given that the $i$th individual message is $w$.
\end{defin}
\begin{defin}
The maximal probability of error for an $(M_i,n_i)$ individual message code
is
\[
\lambda_{\max}(i) = \max_{w \in \mathcal{W}_i} \lambda_{w}(i) \mbox{.}
\]
\end{defin}
\begin{defin}
The maximal probability of error for an $(M_i,n_i)_{i=1}^k$ code is
\[
\lambda_{\max} = \max_{i \in \{1,\ldots,k\}} \lambda_{\max}(i)  \mbox{.}
\]
\end{defin}

Performance criteria weaker than traditional in information theory are defined,
since the Shannon capacity of a channel that dies is zero (Prop.~\ref{prop:cap_zero}).
In particular, we define formal notions of how 
much information is transmitted using a code and how long it takes.
\begin{defin}
The transmission time of an $(M_i,n_i)_{i=1}^k$ code is
$N = \sum_{i=1}^k n_i$.
\end{defin}

\begin{defin}
The expected transmission volume of an $(M_i,n_i)_{i=1}^k$ code is
\[
V = \E_{T} \left\{\sum_{i \in \{ 1,\ldots,k | \hat{W}_i \neq \ominus \}} \log M_i \right\}\mbox{.}
\]
\end{defin}
Notice that although declared erasures do not lead to errors, they do not contribute
transmission volume either.

The several performance criteria for a code may be combined together.
\begin{defin}
Given $0\le \eta < 1$, a pair of numbers $(N_0,V_0)$ (where $N_0$ is a positive integer and $V_0$ is non-negative)
is said to be an \emph{achievable transmission time-volume at $\eta$-reliability} 
if there exists, for some $k$, an $(M_i,n_i)_{i=1}^k$ code for the channel that dies $(\mathcal{X},p_a(y|x),p_d(y|x),p_T(t),\mathcal{Y})$
such that 
\begin{align}
\lambda_{\max} &\le \eta \mbox{, }\\
 N &\le N_0\mbox{, and }\\ 
V &\ge V_0\mbox{.}
\end{align}

Moreover, $(N_0,V_0)$ is said to be an \emph{achievable transmission time-volume 
at Shannon reliability} if it is an achievable transmission time-volume 
at $\eta$-reliability for all $0 < \eta < 1$.
\end{defin}

\section{Limits on Communication}
\label{sec:limcomm}

Having defined the notion of achievable transmission time-volume at various levels of
reliability, the goal of this work is to demarcate what is achievable.

\subsection{Shannon Reliability is Not Achievable}
\label{sec:noShannonRel}

Not only is the Shannon capacity of a channel that dies zero, 
but also there is no $V > 0$ such that $(N,V)$ is an achievable transmission time-volume 
at Shannon reliability.  A coding scheme that always declares erasures would achieve
zero error probability (and therefore Shannon reliability) but would not provide 
positive transmission volume; this is also not allowed under Assumption~\ref{assum:erase}. 

Lemmas are stated and proved after the proof of the main proposition. For
brevity, the proof is limited to the alive-BSC case, but can be extended to general
alive-DMCs by choosing the two most distant letters in $\mathcal{Y}$ for constructing the repetition
code, among other things.
\begin{prop}
\label{prop:noShaRel}
For a channel that dies $(\mathcal{X},p_a(y|x),p_d(y|x),p_T(t),\mathcal{Y})$, 
there is no $V > 0$ such that $(N,V)$ is an achievable transmission time-volume 
at Shannon reliability.
\end{prop}
\begin{IEEEproof}
From the error probability viewpoint, transmitting longer codes
is not harder than transmitting shorter codes (Lem.~\ref{lem:longerbetter}) and
transmitting smaller codes is not harder than transmitting larger codes (Lem.~\ref{lem:smallerbetter}).
Hence, the desired result follows from showing that even the longest and smallest code that has positive expected transmission
volume cannot achieve Shannon reliability.

Clearly the longest and smallest code uses a single individual message code
of length $n_1 \to \infty$ and size $M_1 = 2$.  Among such codes, 
transmitting the binary repetition code is not harder than transmitting
any other code (Lem.~\ref{lem:repetitionbetter}).  Hence showing that the binary
repetition code cannot achieve Shannon reliability yields the desired result.

Consider transmitting a single $(M_1 = 2, n_1)$ individual message code that is simply a 
binary repetition code over a channel that dies $(\mathcal{X},p_a(y|x),p_d(y|x),p_T(t),\mathcal{Y})$.

Let $\mathcal{W}_1 = \{00000\ldots, 11111\ldots \}$, 
where the two codewords are of length $n_1$.  
Assume that the all-zeros codeword and the 
all-ones codeword are each transmitted with probability $1/2$ and measure average probability
of error, since average error probability lower bounds $\lambda_{\max}(1)$ \cite[Problem 5.32]{Gallager1968}.  
The transmission time $N = n_1$ and let $N \to \infty$.  
The expected transmission volume is $\log 2 > 0$.

Under equiprobable signaling over a BSC, the minimum error probability decoder is
the maximum likelihood decoder, which in turn is the minimum distance decoder \cite[Problem 2.13]{McEliece2002}.

The scenario corresponds to binary hypothesis testing over a 
BSC($\varepsilon$) with $T$ observations (since after the channel dies,
the output symbols do not help with hypothesis testing).  Since there is a 
finite $t^{\dagger}$ such that $p_T(t^{\dagger}) > 0$, there is a fixed constant $K$
such that $\lambda_{\max} > K > 0$ for any realization $T = t$.

Thus Shannon reliability is not achievable.
\end{IEEEproof}

\begin{lemma}
\label{lem:longerbetter}
When transmitting over the alive state's memoryless channel $p_a(y|x)$, 
let the maximal probability of error $\lambda_{\max}(i)$ for an optimal $(M_i, n_i)$ individual 
message code and minimum probability of error individual decoder $g_i$ be $\lambda_{\max}(i; n_i)$.  
Then $\lambda_{\max}(i; n_i + 1) \le \lambda_{\max}(i; n_i)$.
\end{lemma}
\begin{IEEEproof}
Consider the optimal block-length-$n_i$ individual message code/decoder, which achieves $\lambda_{\max}(i; n_i)$.
Use it to construct an $n_i + 1$ individual message code that appends a dummy symbol
to each codeword and an associated decoder that operates by ignoring this last symbol.
The error performance of this (suboptimal) code/decoder is clearly $\lambda_{\max}(i; n_i)$, and so
the optimal performance can only be better: $\lambda_{\max}(i; n_i + 1) \le \lambda_{\max}(i; n_i)$.
\end{IEEEproof}

\begin{lemma}
\label{lem:smallerbetter}
When transmitting over the alive state's memoryless channel $p_a(y|x)$, 
let the maximal probability of error $P_e^{\rm max}(i)$ for an optimal $(M_i, n_i)$ individual 
message code and minimum probability of error individual decoder $f_D^{(i)}$ be $P_e^{\rm max}(i; M_i)$.  
Then $P_e^{\rm max}(i; M_i) \le P_e^{\rm max}(i; M_i + 1)$.
\end{lemma}
\begin{IEEEproof}
Follows from sphere-packing principles.
\end{IEEEproof}

\begin{lemma}
\label{lem:repetitionbetter}
When transmitting over the alive state's memoryless channel $p_a(y|x)$, 
the optimal $(M_i = 2, n_i)$ individual message code can be taken as a binary repetition code.
\end{lemma}
\begin{IEEEproof}
Under minimum distance decoding (which yields the minimum error probability \cite[Problem 2.13]{McEliece2002})
for a code transmitted over a BSC, increasing the distance between codewords can only 
reduce error probability.  The repetition code has maximum Hamming distance between codewords.
\end{IEEEproof}

Notice that Prop.~\ref{prop:noShaRel} also directly implies Prop.~\ref{prop:cap_zero}, 
providing an alternate proof.
\begin{cor}
The Shannon capacity of a channel that dies $(\mathcal{X},p_a(y|x),p_d(y|x),p_T(t),\mathcal{Y})$
is zero.
\end{cor}

\subsection{Finite Blocklength Channel Coding}

Before developing an optimal scheme for $\eta$-reliable communication
over a channel that dies, finite block length channel coding is reviewed.

Under our definitions, traditional channel coding results 
\cite{Slepian1963,MacMullenC1998,Laneman2006,PolyanskiyPV2008,BuckinghamV2008,WiechmanS2008}
provide information about individual message codes, determining the achievable trios
$(n_i,M_i,\lambda_{\max}(i))$.  In particular, the largest possible $M_i$ for a 
given $n_i$ and $\lambda_{\max}(i)$ is denoted $M^{*}(n_i,\lambda_{\max}(i))$.

The purpose of this work is not to improve upper and lower bounds on finite
block length channel coding, but to use existing results to study channels that die.
In fact, for the sequel, simply assume that the function $M^{*}(n_i,\lambda_{\max}(i))$ is known, 
as are codes/decoders that
achieve this value.  In principle, optimal individual message codes
may be found through exhaustive search \cite{MacMullenC1998,KaskiO2006}.
Although algebraic notions of code quality do not directly imply 
error probability quality \cite{BargM2005}, perfect codes
such as the Hamming or Golay codes may also be optimal in certain limited
cases.

Recent results comparing upper and lower bounds around Strassen's normal approximation 
to $\log M^{*}(n_i,\lambda_{\max}(i))$ \cite{Strassen1962}
have demonstrated that the approximation is quite good \cite{PolyanskiyPV2008}.
  
\begin{remark}
\label{rem:ass}
We assume that optimal $M^{*}(n_i,\eta)$-achieving individual message codes are known. 
Exact upper and lower bounds to $\log M^{*}(n_i,\eta)$ can be substituted to make our results precise.
For numerical demonstrations, we will further assume that optimal codes have performance
given by Strassen's approximation.
\end{remark}

The following expression for $\log M^{*}(n_i,\eta)$ that first appeared in \cite{Strassen1962} 
is also given as \cite[Thm.~6]{PolyanskiyPV2008}.
\begin{lemma}
\label{lemma:finiteblock}
Let $M^{*}(n_i,\eta)$ be the largest size
of an individual message code with block length $n_i$ and maximal error probability
upper bounded by $\lambda_{\max}(i) < \eta$.  Then, for any DMC with 
capacity $C$ and $0 < \eta \le 1/2$,
\[
\log M^{*}(n_i,\eta) = n_iC - \sqrt{n_i\rho} Q^{-1}(\eta) + O(\log n_i) \mbox{,}
\]
where 
\[
Q(x) = \frac{1}{\sqrt{2\pi}} \int_x^{\infty} e^{-t^2/2} \, dt\mbox{,}
\]
\[
\rho = \min_{X: C = I(X;Y)} \var\left[ \log \frac{p_{Y|X}(y|x)}{p_Y(y)}\right] \mbox{,}
\]
and standard asymptotic notation \cite{Knuth1976} is used.
\end{lemma}

For the BSC($\varepsilon$), the approximation (ignoring the $O(\log n_i)$ term above) is:
\begin{equation}
\label{eq:PPV_BSC}
\log M^{*} \approx n_i(1 - h_2(\varepsilon)) - \sqrt{n_i\varepsilon(1-\varepsilon)} Q^{-1}(\eta)\log_2 \tfrac{\varepsilon}{1-\varepsilon}\mbox{,}
\end{equation}
where $h_2(\cdot)$ is the binary entropy function.
This BSC expression first appeared in \cite{Weiss1960}.

For intuition, we plot the approximate $\log M^{*}(n_i,\eta)$ function
for a BSC($\varepsilon$) in Fig.~\ref{fig:block_vol}.
Notice that $\log M^{*}$ is zero for small $n_i$ since no code can achieve the target
error probability $\eta$.  Also notice that $\log M^{*}$ is a monotonically increasing
function of $n_i$.  Moreover, notice in Fig.~\ref{fig:block_rate} that even when 
normalized, $(\log M^{*})/n_i$, is a monotonically increasing function of $n_i$.  
Therefore longer blocks provide more `bang for the buck.'  The curve in 
Fig.~\ref{fig:block_rate} asymptotically approaches capacity.

\begin{figure}[ht]
  \centering
\subfigure[]{
  \includegraphics[width=3.5in]{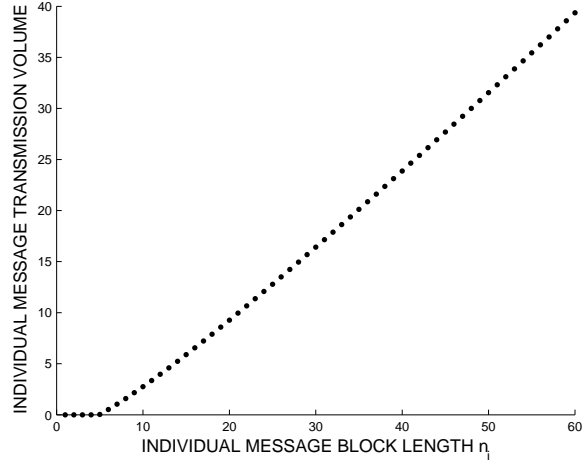}
  \label{fig:block_vol}
	}
\subfigure[]{
  \includegraphics[width=3.5in]{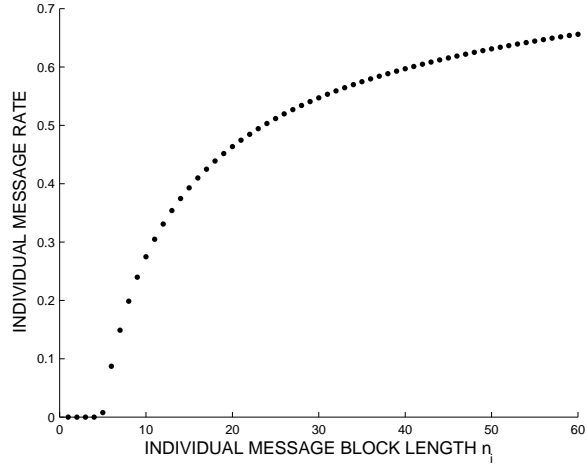} 
  \label{fig:block_rate}
	}
  \caption{\subref{fig:block_vol}. The expression \eqref{eq:PPV_BSC} for $\varepsilon = 0.01$ and $\eta = 0.001$.
	\subref{fig:block_rate}. Normalized version, $(\log M^{*}(n_i,\eta))/n_i$, for $\varepsilon = 0.01$ and $\eta = 0.001$.
	The capacity of a BSC($\varepsilon$) is $1-h_2(\varepsilon) = 0.92$.
  }
\end{figure}

\subsection{$\eta$-reliable Communication}
\label{sec:scheme}
We now describe a coding scheme that achieves positive expected transmission volume at
$\eta$-reliability.  Survival probability of the channel plays a key role in measuring 
performance.
\begin{defin}
The \emph{survival function} of a channel that dies 
$(\mathcal{X},p_a(y|x),p_d(y|x),p_T(t),\mathcal{Y})$
is $\Pr[T > t]$, is denoted $R_T(t)$, and satisfies
\[
R_T(t) = \Pr[T > t] = 1 - \sum_{\tau = 1}^t p_T(\tau) = 1 - F_T(t) \mbox{,}
\]
where $F_T$ is the cumulative distribution function.
\end{defin}
$R_T(t)$ is a non-increasing function.

\begin{prop}
\label{prop:achievability}
The transmission time-volume
\[
\left(N = \sum_{i=1}^k n_i, V = \sum_{i = 1}^k R_T(e_i) \log M^{*}(n_i,\eta) \right)
\]
is achievable at $\eta$-reliability for any sequence $(n_i)_{i=1}^k$ of individual message
codeword lengths, where $e_0 = 0, e_1 = n_1, e_2 = n_1 + n_2,\ldots, e_k = \sum_{i=1}^k n_i$.
\end{prop}
\begin{IEEEproof}

\emph{Code Design:}
A target error probability $\eta$ and a sequence $(n_i)_{i=1}^k$ of individual message
codeword lengths are fixed.  Construct a length-$k$ sequence of $(M_i,n_i)$ 
individual message codes and individual decoding functions 
$(\mathcal{W}_i,f_i,g_i)$ that achieve optimal performance.
The size of $\mathcal{W}_i$ is $|\mathcal{W}_i| = \log M^{*}(n_i,\eta)$.  Note that 
individual decoding functions $g_i$ have range $\mathcal{W}_i$ rather than 
$\mathcal{W}_i^{\ominus}$.

\emph{Encoding:}
A codeword $W_1 = w_1$ is selected uniformly at random from the codebook 
$\mathcal{W}_1$.  The mapping of this codeword into $n_1$ channel input letters,
$X_{e_0 + 1}^{e_1} = f_1(w_1)$, is transmitted in channel usage times $n = e_0 + 1, e_0 + 2,\ldots,e_1$.

Then a codeword $W_2 = w_2$ is selected uniformly at random from the codebook 
$\mathcal{W}_2$.  The mapping of this codeword into $n_2$ channel input letters,
$X_{e_1 + 1}^{e_2} = f_2(w_2)$, is transmitted in channel usage times $n = e_1+1,e_1+2,\ldots,e_2$.

This procedure continues until the last individual message code in the code is transmitted.
That is, a codeword $W_k = w_k$ is selected uniformly at random from the codebook 
$\mathcal{W}_k$.  The mapping of this codeword into $n_k$ channel input letters,
$X_{e_{k-1} + 1}^{e_k} = f_k(w_k)$, is transmitted in channel usage times 
$n = e_{k-1}+ 1, e_{k-1}+ 2,\ldots,e_k$.

We refer to channel usage times $n \in \{e_{i-1} + 1,e_{i-1} + 2,\ldots,e_i\}$ as the $i$th transmission epoch.  

\emph{Decoding:}
For decoding, the channel output symbols for each epoch are processed separately.
If any of the channel output symbols in an epoch are erasure symbols $?$, 
then a decoding erasure $\ominus$ is declared for the message in that epoch, 
i.e.\ $\hat{W}_i = \ominus$.  Otherwise, the individual
message decoding function $g_i: \mathcal{Y}^{n_i} \to \mathcal{W}_i$ is applied
to obtain $\hat{W}_i = g_i(Y_{e_{i-1} + 1}^{e_i})$.

\emph{Performance Analysis:}
Having defined the communication scheme, we measure the error probability, transmission time, 
and expected transmission volume.

The decoder will either produce an erasure $\ominus$ or use an individual 
message decoder $g_i$.  When $g_i$ is used, the maximal error 
probability of individual message code error is bounded as
$\lambda_{\max}(i) < \eta$ 
by construction.  Since declared erasures $\ominus$ do not lead to error, and since all
$\lambda_{\max}(i) < \eta$, it follows that
\[
\lambda_{\max} < \eta \mbox{.}
\]

The transmission time is simply $N = \sum n_i$. 

Recall the definition of expected transmission volume:
\[
\E \left\{\sum_{i \in \{ 1,\ldots,k | \hat{W}_i \neq \ominus \}} \log M_i \right\} = \sum_{i \in \{ 1,\ldots,k | \hat{W}_i \neq \ominus \}} \E\left\{ \log M_i \right\}
\]
and the fact that the channel produces the erasure symbol $?$ for all channel usage
times after death, $n > T$, but not before.  Combining this with the length of an optimal code, $\log M^{*}(n_i,\eta)$,
leads to the expression
\[
\sum_{i = 1}^k \Pr[T > e_i] \log M^{*}(n_i,\eta) \mbox{,}
\]
since all individual message codewords that are received in their entirety
before the channel dies are decoded using $g_i$ whereas any
individual message codewords that are even partially cut off are declared $\ominus$.

Recalling the definition of the survival function,
the expected transmission volume of the communication scheme is
\[
\sum_{i = 1}^k R_T(e_i) \log M^{*}(n_i,\eta)
\] 
as desired.
\end{IEEEproof}

Prop.~\ref{prop:achievability} is valid for any choice of $(n_i)_{i=1}^k$.
Since $(\log M^{*})/n_i$ is monotonically increasing, it is better
to use individual message codes that are as long as possible.  With longer individual message
codes, however, there is a greater chance of many channel usages being wasted if
the channel dies in the middle of transmission.  The basic trade-off is captured in picking
the set of values $\{n_1,n_2,\ldots,n_k\}$.  For fixed and finite $N$, this involves
picking an ordered integer partition $n_1 + n_2 + \cdots + n_k = N$.  
We optimize this choice in Section~\ref{sec:optim}.

\subsection{Converse Arguments}
\label{sec:converseargs}
Since we simply have operational expressions and no informational expressions
in our development, as per Remark~\ref{rem:ass}, and since optimal individual message codes and individual message 
decoders are assumed to be used, it may seem as though converse arguments are not required.  
This would indeed follow, if the following two things were true, which follow
from Assumption~\ref{assum:erase}.  First, that there is no benefit
in trying to decode the last partially erased message block.  Second, that
there is no benefit to errors-and-erasures decoding \cite{Forney1968} by the $g_i$ for codewords 
that are received before channel death.  
Under Assumption~\ref{assum:erase}, Prop.~\ref{prop:achievability}
gives the best performance possible.

One might wonder whether Assumption~\ref{assum:erase} is needed.
That there would be no benefit in trying to decode the last partially erased block 
follows from the conjecture that an optimal individual message code would have no 
latent redundancy that could be exploited to achieve a $\lambda_{\max}(i = \rm{last}) < \eta$,
but this is a property of the actual optimal code.

Understanding the possibility of errors-and-erasures decoding \cite{Forney1968}
by the individual message decoders also requires knowing properties of actual optimal codes.  
It is unclear how the choice of threshold in errors-and-erasures decoding 
would affect the expected transmission volume
\[
\sum_{i = 1}^k (1-\xi_i) R_T(e_i) \log M^{*}(n_i,\xi_i,\eta) \mbox{,}
\]
where $\xi_i$ would be the specified erasure probability for individual message $i$,
and $M^{*}(n_i,\xi_i,\eta)$ would be the maximum individual message codebook size
under erasure probability $\xi_i$ and maximum error probability $\eta$.  

What we can say, however, is that at the level of Strassen's approximation (up to the $\log n$ term), 
$\log M^{*}(n_i,\xi_i,\eta)$ and $\log M^{*}(n_i,\eta)$ are the same \cite[Thm.~47]{Polyanskiy2010}.

\section{Optimizing the Communication Scheme}
\label{sec:optim}

In Section~\ref{sec:scheme}, we had not optimized the lengths of the
individual message codes; we do so here.  For fixed $\eta$ and $N$, we maximize the expected 
transmission volume $V$ over the choice of the ordered integer partition $n_1 + n_2 + \cdots + n_k = N$:
\begin{equation}
\label{eq:optim}
\max_{(n_i)_{i=1}^k: \sum n_i = N}  \sum_{i = 1}^k R_T(e_i) \log M^{*}(n_i,\eta) \mbox{.}
\end{equation}

For finite $N$, this optimization can be carried out by an exhaustive search
over all $2^{N-1}$ ordered integer partitions.  If the death distribution $p_T(t)$
has finite support, there is no loss of generality in considering only
finite $N$.  Since exhaustive search has exponential complexity, however, there 
is value in trying to use a simplified algorithm.  A dynamic programming
formulation for the finite horizon case is developed in Section~\ref{sec:DP}.
The next subsection develops a greedy algorithm which is applicable
to both the finite and infinite horizon cases and yields the optimal
solution for certain problems.

\subsection{A Greedy Algorithm}

To try to solve the optimization problem \eqref{eq:optim}, we propose a
greedy algorithm that optimizes blocklengths $n_i$ one by one.
\begin{algorithm}
\label{algo:greed}
\quad
\begin{enumerate}
  \item Maximize $R_T(n_1) \log M^{*}(n_1,\eta)$ through the choice of $n_1$ independently of any other $n_i$.  
  \item Maximize $R_T(e_2) \log M^{*}(n_2,\eta)$ after fixing $e_1 = n_1$, but independently of later $n_i$.  
  \item Maximize $R_T(e_3) \log M^{*}(n_3,\eta)$ after fixing $e_2$, but independently of later $n_i$.
  \item Continue in the same manner for all subsequent $n_i$.
\end{enumerate}
\end{algorithm}

Sometimes the algorithm produces the correct solution.
\begin{prop}
\label{prop:localopt}
The solution produced by the greedy algorithm, $(n_i)$, is locally optimal if
\begin{equation}
\label{eq:star}
\frac{R_T(e_i)\log M^{*}(n_i,\eta) - R_T(e_i-1)\log M^{*}(n_i - 1,\eta)}{R_T(e_{i+1})\left[\log M^{*}(n_{i+1} + 1,\eta) - \log M^{*}(n_{i+1},\eta)\right]} \ge 1 
\end{equation}
for each $i$.
\end{prop}
\begin{IEEEproof}
The solution of the greedy algorithm partitions time using a set of epoch boundaries $(e_i)$.
The proof proceeds by testing whether local perturbation of an arbitrary 
epoch boundary can improve performance.  There are two possible perturbations: a shift to the left
or a shift to the right.

First consider shifting an arbitrary epoch boundary $e_i$ to the right by one.  This makes the
left epoch longer and the right epoch shorter.  Lengthening the left epoch does not improve performance
due to the greedy optimization of the algorithm.  Shortening the right epoch does not improve
performance since $R_T(e_i)$ remains unchanged whereas $\log M^{*}(n_i,\eta)$ does not increase since
$\log M^{*}$ is a non-decreasing function of $n_i$.

Now consider shifting an arbitrary epoch boundary $e_i$ to the left by one.  This
makes the left epoch shorter and the right epoch longer.  Reducing the left epoch will
not improve performance due to greediness, but enlarging the right epoch might improve performance, 
so the gain and loss must be balanced.

The loss in performance (a positive quantity) for the left epoch is
\[
\Delta_l = R_T(e_i)\log M^{*}(n_i,\eta) - R_T(e_i-1)\log M^{*}(n_i - 1,\eta)
\]
whereas the gain in performance (a positive quantity) for the right epoch is
\[
\Delta_r = R_T(e_{i+1})\left[\log M^{*}(n_{i+1} + 1,\eta) - \log M^{*}(n_{i+1},\eta)\right] \mbox{.}
\]
If $\Delta_l \ge \Delta_r$, then perturbation will not improve performance.
The condition may be rearranged as 
\[
\frac{R_T(e_i)\log M^{*}(n_i,\eta) - R_T(e_i-1)\log M^{*}(n_i - 1,\eta)}{R_T(e_{i+1})\left[\log M^{*}(n_{i+1} + 1,\eta) - \log M^{*}(n_{i+1},\eta)\right]} \ge 1 
\]
This is the condition \eqref{eq:star}, so the left-perturbation does not improve
performance.
Hence, the solution produced by the greedy algorithm is locally optimal.
\end{IEEEproof}
\begin{prop}
The solution produced by the greedy algorithm, $(n_i)$, is globally optimal if
\begin{equation}
\label{eq:starK}
\frac{R_T(e_i)\log M^{*}(n_i,\eta) - R_T(e_i-K_i)\log M^{*}(n_i - K_i,\eta)}{R_T(e_{i+1})\left[\log M^{*}(n_{i+1} + K_i,\eta) - \log M^{*}(n_{i+1},\eta)\right]} \ge 1 
\end{equation}
for each $i$, and any non-negative integers $K_i \le n_i$.
\end{prop}
\begin{IEEEproof}
The result follows by repeating the argument for local optimality in Prop.~\ref{prop:localopt}
for shifts of any admissible size $K_i$. 
\end{IEEEproof}

There is an easily checked special case of global optimality condition \eqref{eq:starK} under
the Strassen approximation, given in the forthcoming Prop.~\ref{prop:epochsizeorder}.
\begin{lemma}
\label{lem:nondecr}
The function $\log M^{*}_{S}(z,\eta) - \log M^{*}_{S}(z-K,\eta)$ is a non-decreasing
function of $z$ for any $K$, where
\begin{equation}
\label{eq:Strass}
\log M^{*}_{S}(z,\eta) = zC - \sqrt{z\rho} Q^{-1}(\eta)
\end{equation}
is Strassen's approximation.
\end{lemma}
\begin{IEEEproof}
Essentially follows from the fact that $\sqrt{z}$ is a concave $\cap$
function in $z$.  More specifically $\sqrt{z}$ satisfies
\[
-\sqrt{z} + \sqrt{z - K} \le -\sqrt{z+1} + \sqrt{z+1-K}
\]
for $K \le z$.  This implies:
\[
- \sqrt{z}\sqrt{\rho} Q^{-1}(\eta) + \sqrt{z-K}\sqrt{\rho} Q^{-1}(\eta) \le - \sqrt{z+1}\sqrt{\rho} Q^{-1}(\eta) + \sqrt{z+1-K}\sqrt{\rho} Q^{-1}(\eta)\mbox{.}
\]
Adding the positive constant $KC$ to both sides, in the form $zC - zC + KC$ on the left and in the form $(z+1)C - (z+1)C + KC$ on the right
yields
\begin{align*}
&zC - \sqrt{z\rho} Q^{-1}(\eta) - (z-K)C + \sqrt{z-K}\sqrt{\rho} Q^{-1}(\eta) \\ \notag
&\quad \le (z+1)C - \sqrt{z+1}\sqrt{\rho} Q^{-1}(\eta) - (z+1-K)C + \sqrt{z+1-K}\sqrt{\rho} Q^{-1}(\eta)
\end{align*}
and so
\[
\left[\log M^{*}_{S}(z,\eta) - \log M^{*}_{S}(z-K,\eta)\right] \le \left[\log M^{*}_{S}(z+1,\eta) - \log M^{*}_{S}(z+1-K,\eta)\right]\mbox{.}
\]
\end{IEEEproof}

\begin{prop}
\label{prop:epochsizeorder}
If the solution produced by the greedy algorithm using Strassen's approximation
\eqref{eq:Strass} satisfies $n_1 \ge n_2 \ge \cdots \ge n_k$,
then condition \eqref{eq:starK} for global optimality is satisfied.
\end{prop}
\begin{IEEEproof}
Since $R_T(\cdot)$ is a non-increasing survival function, 
\begin{equation}
\label{eq:survival}
R_T(e_i - K) \ge R_T(e_{i+1})
\end{equation}
for the non-negative integer $K$.
Since the function $\left[\log M^{*}_{S}(z,\eta) - \log M^{*}_{S}(z-K,\eta)\right]$ is a non-decreasing
function of $z$ by Lem.~\ref{lem:nondecr}, and since the $n_i$ are in non-increasing order,
\begin{equation}
\label{eq:monot}
\log M^{*}_{S}(n_i,\eta) - \log M^{*}_{S}(n_i - K,\eta) \ge \log M^{*}_{S}(n_{i+1} + K,\eta) - \log M^{*}_{S}(n_{i+1},\eta)\mbox{.}
\end{equation}
Taking products of \eqref{eq:survival} and \eqref{eq:monot} and 
rearranging yields the condition:
\[
\frac{R_T(e_i - K)\left[ \log M^{*}_{S}(n_i,\eta) - \log M^{*}_{S}(n_i - K,\eta) \right]}{R_T(e_{i+1})\left[ \log M^{*}_{S}(n_{i+1} + K,\eta) - \log M^{*}_{S}(n_{i+1},\eta)\right]} \ge 1 \mbox{.}
\]
Since $R_T(\cdot)$ is a non-increasing survival function,
\[
R_T(e_i - K) \ge R_T(e_i) \ge R_T(e_{i+1}) \mbox{.}
\]
Therefore the global optimality condition \eqref{eq:starK} is also satisfied, by substituting $R_T(e_i)$ for $R_T(e_i - K)$
in one place.
\end{IEEEproof}

\subsection{Geometric Death Distribution}
A common failure mode for systems that do not age 
is a geometric death time $T$ \cite{Davis1952}:
\[
p_T(t) = \alpha(1-\alpha)^{(t-1)} \mbox{,}
\]
and
\[
R_T(t) = (1-\alpha)^t \mbox{,}
\]
where $\alpha$ is the death time parameter.

\begin{prop}
\label{prop:equalsize}
When $T$ is geometric, then the solution to \eqref{eq:optim} 
under Strassen's approximation yields equal epoch sizes.
This optimal size is given by
\[
\argmax_{\nu} R_T(\nu) \log M^{*}(\nu,\eta) \mbox{.}
\] 
\end{prop}
\begin{IEEEproof}
Begin by showing that Algorithm~\ref{algo:greed} will produce a solution with equal epoch sizes.
Recall that the survival function of a geometric random variable with parameter $0 < \alpha \le 1$ is
$R_T(t) = (1-\alpha)^t$.  Therefore the first step of the algorithm will choose $n_1$ as
\[
n_1 = \argmax_{\nu} (1-\alpha)^{\nu} \log M^{*}(\nu,\eta) \mbox{.}
\]
The second step of the algorithm will choose
\begin{align*}
n_2 &= \argmax_{\nu} (1-\alpha)^{n_1}(1-\alpha)^{\nu} \log M^{*}(\nu,\eta) \\ \notag
&= \argmax_{\nu} (1-\alpha)^{\nu} \log M^{*}(\nu,\eta) \mbox{,}
\end{align*}
which is the same as $n_1$.  In general,
\begin{align*}
n_i &= \argmax_{\nu} (1-\alpha)^{e_{i-1}}(1-\alpha)^{\nu} \log M^{*}(\nu,\eta) \\ \notag
&= \argmax_{\nu} (1-\alpha)^{\nu} \log M^{*}(\nu,\eta) \mbox{,}
\end{align*}
so $n_1 = n_2 = \cdots$.

Such a solution satisfies $n_1 \ge n_2 \ge \cdots$ 
and so it is optimal by Prop.~\ref{prop:epochsizeorder}.
\end{IEEEproof}

The optimal epoch size for geometric death under Strassen's approximation can
be found analytically, \cite[Sec.~6.4.2]{Varshney2010}.  Consider
the setting when the alive state corresponds to a BSC($\varepsilon$). 
For fixed crossover probability $\varepsilon$ and target error probability $\eta$, the
optimal epoch size is plotted as a function of $\alpha$ in Fig.~\ref{fig:blocksizes}. The less likely the channel is to
die early, the longer the optimal epoch length.
\begin{figure}
  \centering
  \includegraphics[width=3.5in]{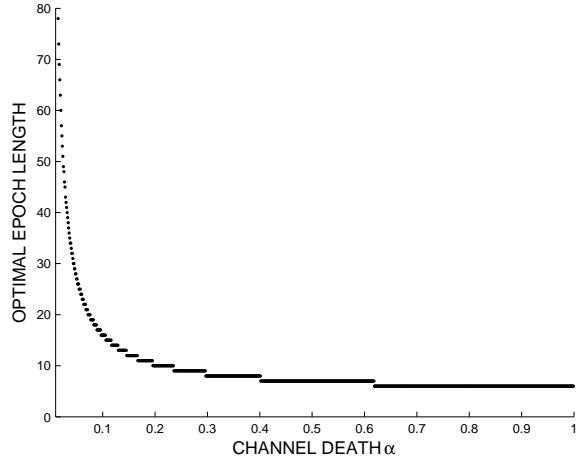}
  \caption{Optimal epoch lengths under Strassen's approximation
	for an $(\varepsilon,\alpha)$ BSC-geometric channel that dies for $\varepsilon = 0.01$ and $\eta = 0.001$.}
  \label{fig:blocksizes}
\end{figure}

Alternatively, rather than fixing $\eta$, one might fix the number of bits to be communicated
and find the best level of reliability that is possible.  Fig.~\ref{fig:bscgeom} shows the best $\lambda_{\max} = \eta$
that is possible when communicating $5$ bits over a BSC($\varepsilon$)-geometric($\alpha$) channel that dies.
\begin{figure}
  \centering
  \includegraphics[width=3.5in]{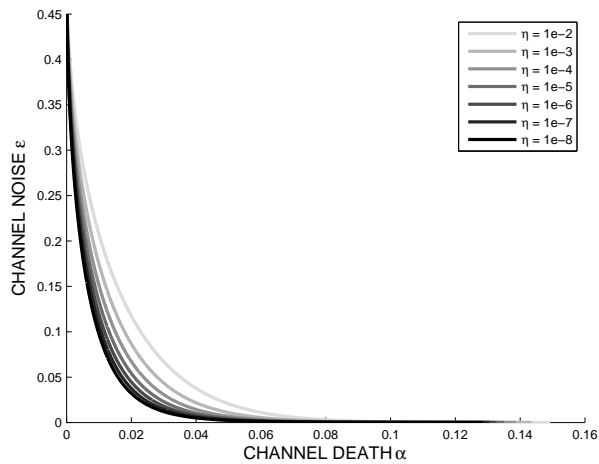}
  \caption{Achievable $\eta$-reliability in sending $5$ bits 
	over $(\varepsilon,\alpha)$ BSC-geometric channel that dies.}
  \label{fig:bscgeom}
\end{figure}

Notice that the geometric death time distribution
forms a boundary case for Prop.~\ref{prop:epochsizeorder}.
One can consider discrete Weibull death time distributions \cite{KhanKA1989}
to see what happens with heavier tails:
\[
p_T(t) = (1-\alpha)^{(t-1)^{\beta}} - (1-\alpha)^{t^{\beta}} \mbox{,}
\]
and
\[
R_T(t) = (1-\alpha)^{t^{\beta}} \mbox{,}
\]
where $\beta$ is the shape parameter.  When $\beta > 1$, the tail is 
lighter than geometric and when $\beta < 1$, the tail is heavier than
geometric.  

With heavy-tailed death distributions, the greedy algorithm gives
epoch sizes that are non-increasing: $n_1 \ge n_2 \ge \cdots$, and 
therefore optimal;
it is better to send long blocks first and then send shorter ones.

\subsection{Dynamic Programming}
\label{sec:DP}
The greedy algorithm of the previous section solves \eqref{eq:optim}
under certain conditions.  For finite $N$, a dynamic program (DP) 
may be used to solve \eqref{eq:optim} under any conditions.
To develop the DP formulation \cite{Bertsekas2005}, we assume that  
channel state feedback (whether the channel output is $?$ or whether it is some other symbol)
is available to the transmitter, however solving the DP will show that channel 
state feedback is not required.  

\emph{System Dynamics:}
\begin{equation}
\label{eq:dynamics}
\begin{bmatrix} \zeta_n \\ \omega_n \end{bmatrix} = \begin{bmatrix} (\zeta_{n-1} + 1)\hat{s}_{n-1} \\ \omega_{n-1}\kappa_{n-1} \end{bmatrix} \mbox{,}
\end{equation}
for $n = 1,2,\ldots,N+1$.  The following state variables, disturbances, and controls are used:
\begin{itemize}
  \item $\zeta_n \in \mathbb{Z}^{*}$ is a state variable that counts the location in the current transmission epoch,
  \item $\omega_n \in \{0,1\}$ is a state variable that indicates whether the channel is alive ($1$) or dead ($0$),
  \item $\kappa_n \in \{0,1\} \sim \Bern\left( R_T(n)\right)$ is a disturbance that kills ($0$) or revives ($1$) the channel in the next time step, and
  \item $\hat{s}_n \in \{0,1\}$ is a control input that starts ($0$) or continues ($1$) a transmission epoch in the next time step.
\end{itemize}

\emph{Initial State:}
Since the channel starts alive (note that $R_T(1) = 1$) and since the first transmission epoch starts at the beginning of time, 
\begin{equation}
\label{eq:initialstate}
\begin{bmatrix} \zeta_1 \\ \omega_1\end{bmatrix} = \begin{bmatrix} 0 \\ 1 \end{bmatrix} \mbox{.}
\end{equation}

\emph{Additive Cost:} Transmission volume $\log M^{*}(\zeta_{n} + 1,\eta)$ is credited if the channel is alive (i.e.\ $\omega_n = 1$)
and the transmission epoch is to be restarted in the next time step (i.e.\ $1-\hat{s}_n = 1$).  This implies a cost function
\begin{equation}
\label{eq:additivecost}
c_n(\zeta_n,\omega_n,\hat{s}_n) = -(1-\hat{s}_n)\omega_n \log M^{*}(\zeta_n + 1,\eta) \mbox{.}
\end{equation}
This is negative so that smaller is better.

\emph{Terminal Cost:} There is no terminal cost: $c_{N+1} = 0$.

\emph{Cost-to-go:} From time $n$ to time $N+1$ is:
\[
\E_{\vec{\kappa}}\left\{ \sum_{i=n}^N c_i(\zeta_i,\omega_i,\hat{s}_i)\right\}
= -\E_{\vec{\kappa}}\left\{ \sum_{i=n}^N (1-\hat{s}_i)\omega_i \log M^{*}(\zeta_i + 1,\eta)\right\} \mbox{.}
\]

Notice that the state variable $\zeta_n$ which counts epoch time is known to the transmitter and 
is determinable by the receiver through transmitter simulation.  The state variable $\omega_n$ indicates 
the channel state and is known to the receiver by observing the channel output. It may be communicated
to the transmitter through the channel state feedback.  The following result follows directly.
\begin{prop}
A communication scheme that follows the dynamics \eqref{eq:dynamics} and additive cost \eqref{eq:additivecost} 
achieves the transmission time-volume
\[
\left(N, V = -\E\left[ \sum_{n = 1}^N c_n \right]\right)
\]
at $\eta$-reliability.
\end{prop}

DP may be used to find the optimal control policy $(\hat{s}_n)$.
\begin{prop}
\label{prop:DP}
The optimal $-V$ for the initial state \eqref{eq:initialstate}, dynamics \eqref{eq:dynamics},
additive cost \eqref{eq:additivecost}, and no terminal cost
is equal to the cost of the solution produced by the dynamic programming algorithm.
\end{prop}
\begin{IEEEproof}
The system described by initial state \eqref{eq:initialstate}, dynamics \eqref{eq:dynamics},
and additive cost \eqref{eq:additivecost} is in the form of the \emph{basic problem} of
dynamic programming \cite[Sec.~1.2]{Bertsekas2005}.
Thus the result follows from \cite[Prop.~1.3.1]{Bertsekas2005}
\end{IEEEproof}

The DP optimization computations are now carried out; standard $J$ notation is
used for cost \cite{Bertsekas2005}.
The base case at time $N+1$ is 
\[
J_{N+1}(\zeta_{N+1},\omega_{N+1}) = c_{N+1} = 0 \mbox{.}
\]
In proceeding backwards from time $N$ to time $1$:
\[
J_n(\zeta_n,\omega_n) = \min_{\hat{s}_n\in\{0,1\}} \E_{\kappa_n}\left\{ c_n(\zeta_n,\omega_n,\hat{s}_n) + J_{n+1}\left(f_n(\zeta_n,\omega_n,\hat{s}_n,\kappa_n)\right) \right\} \mbox{,}
\]
for $n = 1,2,\ldots,N$, where 
\begin{align*}
f_n(\zeta_n,\omega_n,\hat{s}_n,\kappa_n) &= \begin{bmatrix}\zeta_{n+1}&\omega_{n+1}\end{bmatrix}^T \\
&= \begin{bmatrix}(\zeta_n + 1)\hat{s}_n&\omega_n \kappa_n\end{bmatrix}^T\mbox{.}
\end{align*}
Substituting our additive cost function yields:
\begin{align}
\label{eq:dropout}
J_n(\zeta_n,\omega_n) &= \min_{\hat{s}_n\in\{0,1\}} -\E_{\kappa_n}\left\{ (1-\hat{s}_n)\omega_n\log M^{*}(\zeta_n + 1,\eta)\right\} + \E_{\kappa_n}\{J_{n+1}\} \\ \notag
&= \min_{\hat{s}_n\in\{0,1\}} -(1-\hat{s}_n)R_T(n)\log M^{*}(\zeta_n + 1,\eta) + \E_{\kappa_n}\{J_{n+1}\} \mbox{.}
\end{align}
Notice that the state variable $\omega_n$ dropped out of the first term when we took the expectation
with respect to the disturbance $\kappa_n$.  This is true for each stage in the DP.
\begin{prop}
For a channel that dies $(\mathcal{X},p_a(y|x),p_d(y|x),p_T(t),\mathcal{Y})$, channel state feedback does not improve performance.
\end{prop}
\begin{IEEEproof}
By repeating the expectation calculation in \eqref{eq:dropout} for each stage
$n$ in the stage-by-stage DP algorithm, it is verified that state variable $\omega$ does not enter into the 
stage optimization problem.  Hence the transmitter does not require channel state feedback to 
determine the optimal signaling strategy.
\end{IEEEproof}

\subsection{A Dynamic Programming Example}
\label{sec:example}
To provide some intuition on the choice of epoch lengths, we present a short example.
Consider the channel that dies with $\mathcal{X} = \{0,1\}$, $\mathcal{Y} = \{0,1,?\}$,
$p_a(y|x)$ given by \eqref{eq:pa} with $\varepsilon = 0.01$, $p_d(y|x)$ given by \eqref{eq:pd},
and $p_T(t)$ that is uniform over a finite horizon of length $40$ (disallowing death 
in the first time step):
\[
p_T(t) = \begin{cases} 1/39, & t = 2,\ldots,40, \\ 0 & \mbox{otherwise.}\end{cases}
\]
Our goal is to communicate with $\eta$-reliability, $\eta = 0.001$.

Since the death distribution has finite support, there is no benefit to transmitting
after death is guaranteed.  Suppose some sequence of $n_i$s is chosen arbitrarily:
$(n_1 = 13, n_2 = 13, n_3 = 13, n_4 = 1)$.  This has expected transmission volume (under the Strassen approximation)
\begin{align*}
V &= \sum_{i=1}^4 R_T(e_i) \log M^{*}(n_i,\eta)  \\ \notag
&\stackrel{(a)}{=} \log M^{*}(13,0.001) \sum_{i=1}^3 R_T(e_i)  \\ \notag
&= \log M^{*}(13,0.001)[R_T(13) + R_T(26)+ R_T(39)] \\ \notag
&= 4.600[9/13 + 14/39 + 1/39] = 4.954 \mbox{ bits.}
\end{align*}
where (a) removes the fourth epoch since uncoded transmission cannot achieve $\eta$-reliability.

If we run the DP algorithm to optimize the ordered integer partition, we
get the result $(n_1 = 20, n_2 = 12, n_3 = 6, n_4 = 2)$.\footnote{Equivalently 
$(n_1 = 20, n_2 = 12, n_3 = 6, n_4 = 1, n_5 = 1)$, since the last two channel usages
are wasted (see Fig.~\ref{fig:block_vol}) to hedge against channel death.}  
Notice that since the solution is in order, the greedy algorithm would also 
have succeeded.  The expected transmission volume for this strategy (under the Strassen approximation) is
\begin{align*}
V &= R_T(20)\log M^{*}(20,0.001) + R_T(32)\log M^{*}(12,0.001) + R_T(38)\log M^{*}(6,0.001) \\ \notag
&= (20/39)\cdot 9.2683 + (8/39)\cdot 3.9694 + (2/39)\cdot 0.5223 \\ \notag
&= 5.594 \mbox{ bits.}
\end{align*}

\subsection{A Precise Solution}
It has been assumed that optimal finite block length codes
are known and used.  Moreover, the Strassen approximation has been used for certain
computations.  It is, however, also of interest to determine precisely which code should be 
used over a channel that dies.  This subsection gives an example where a sequence of length-$23$ binary Golay
codes \cite{Golay1949} are optimal.  Similar examples may be developed for other perfect codes;
a perfect code is one for which there are equal-radius spheres centered at the
codewords that are disjoint and that completely fill $\mathcal{X}^{n_i}$.

Before presenting the example, the sphere-packing upper bound on $\log M^{*}(n_i,\eta)$ for a
BSC($\varepsilon$) is derived.  Recall the notion of decoding radius \cite{Blahut1983}
and let $\rho(\varepsilon,\eta)$ be the largest integer such that
\[
\sum_{s=0}^{\rho} \binom{n_i}{s} \varepsilon^{s} (1-\varepsilon)^{n_i - s} \le 1 - \eta \mbox{.}
\]
The sphere-packing bound follows from counting how many decoding regions
of radius $\rho$ could conceivably fit in the Hamming space $2^{n_i}$ disjointly.
Let $D_{s,m}$ be the number of channel output sequences that are decoded into message $w_m$
and have distance $s$ from the $m$th codeword.  By the nature of Hamming space,
\[
D_{s,m} \le \binom{n_i}{s}
\]
and due to the volume constraint,
\[
\sum_{m=1}^M \sum_{s = 0}^{\rho} D_{s,m} \le 2^{n_i} \mbox{.}
\]
Hence, the maximal codebook size $M^{*}(n_i,\eta)$ is upper-bounded as
\begin{align*}
M^{*}(n_i,\eta) &\le \frac{2^{n_i}}{\sum_{s=0}^{\rho} D_{s,m} } \\ \notag
	&\le \frac{2^{n_i}}{\sum_{s=0}^{\rho(\varepsilon,\eta)} \binom{n_i}{s} } \mbox{.}
\end{align*}
Thus the sphere-packing upper bound on $\log M^{*}(n_i,\eta)$ is
\[
\log M^{*}(n_i,\eta) \le n_i - \log\left[ \sum_{s=0}^{\rho(\varepsilon,\eta)} \binom{n_i}{s} \right] \triangleq \log M_{sp}(n_i,\eta)  \mbox{.}
\]
Perfect codes such as the binary Golay code
of length $23$ can sometimes achieve the sphere-packing bound with equality. 

Consider an $(\varepsilon,\alpha)$ BSC-geometric channel that dies, with $\varepsilon = 0.01$ and $\alpha = 0.05$.
The target error probability is fixed at $\eta = 2.9 \times 10^{-6}$.  For these values of $\varepsilon$ and $\eta$,
the decoding radius $\rho(\varepsilon,\eta) = 1$ for $2 \le n_i \le 3$.  It is $\rho(\varepsilon,\eta) = 2$ for $4 \le n_i \le 10$;
$\rho(\varepsilon,\eta) = 3$ for $11 \le n_i \le 23$; $\rho(\varepsilon,\eta) = 4$ for $24 \le n_i \le 40$; and so on.  

Moreover, one can note that the $(n = 23, M = 4096)$ binary Golay code has a decoding radius of $3$; thus it meets the 
BSC sphere-packing bound 
\[
M_{sp}(23,2.9 \times 10^{-6}) = \frac{2^{23}}{1 + 23 + 253 + 1771} = 4096
\]
with equality.  

Now to bring channel death into the picture.  If one proceeds greedily, following Algorithm~\ref{algo:greed},
but using the sphere-packing bound $\log M_{sp}(n_i,\eta)$ rather than the optimal $\log M^{*}(n_i,\eta)$, 
\begin{align*}
&n_1(\varepsilon = 0.01,\alpha = 0.05, \eta = 2.9 \times 10^{-6}) \\
&\quad = \argmax_{\nu} \bar{\alpha}^{\nu} \log_2 \frac{2^{\nu}}{\sum_{s=0}^{\rho(\varepsilon,\eta)}} = 23 \mbox{.}
\end{align*}
By the memorylessness argument of Prop.~\ref{prop:equalsize}, it follows that
running Algorithm~\ref{algo:greed} with the sphere-packing bound will yield $23 = n_1 = n_2 = \cdots$.

It remains to show that Algorithm~\ref{algo:greed} actually gives the true solution.  Had Strassen's
approximation been used rather than the sphere-packing bound, the result would follow directly from Prop.~\ref{prop:equalsize}.
Instead, the global optimality condition \eqref{eq:starK} can be verified exhaustively for 
all $23$ possible shift sizes $K$ for the first epoch:
\[
\frac{\bar{\alpha}^{23} \log M_{sp}(23,\eta) - \bar{\alpha}^{23-K}\log M_{sp}(23-K,\eta)}{\bar{\alpha}^{46}\log M_{sp}(23+K) - \bar{\alpha}^{46}\log M_{sp}(23,\eta)} \ge 1 \mbox{.}
\]
Then the same exhaustive verification is performed
for all $23$ possible shifts for the second epoch:
\begin{align*}
\frac{ \bar{\alpha}^{46} \log M_{sp}(23,\eta) - \bar{\alpha}^{46-K}\log M_{sp}(23-K,\eta) }{ \bar{\alpha}^{69}\log M_{sp}(23+K) - \bar{\alpha}^{69}\log M_{sp}(23,\eta)} &\ge 1 \\ \notag
\frac{ \bar{\alpha}^{23} \left[\bar{\alpha}^{23} \log M_{sp}(23,\eta) - \bar{\alpha}^{23-K}\log M_{sp}(23-K,\eta)\right] }{ \bar{\alpha}^{23} \left[\bar{\alpha}^{46}\log M_{sp}(23+K) - \bar{\alpha}^{46}\log M_{sp}(23,\eta)\right]} &\ge 1 \\ \notag
\frac{ \bar{\alpha}^{23} \log M_{sp}(23,\eta) - \bar{\alpha}^{23-K}\log M_{sp}(23-K,\eta) }{ \bar{\alpha}^{46}\log M_{sp}(23+K) - \bar{\alpha}^{46}\log M_{sp}(23,\eta)} &\ge 1 \mbox{.}
\end{align*}
The exhaustive verification can be carried out indefinitely to show that using the length-$23$
binary Golay code for every epoch is optimal.

\subsection{Practical Codes and Empirical Death Distributions}
It should be noted that the algorithms developed for optimizing communication schemes over channels that die
work with arbitrary death distributions, even empirically measured ones, e.g.\ the experimentally 
characterized death properties of a synthetic biology communication system \cite[Fig.~3: Reliability]{CantonLE2008}.

Further, rather than considering the $\log M^{*}(n_i, \eta)$ function for optimal finite block length 
codes, the code optimization procedures would work just as well if a collection of finite block length codes was
provided. Such a limited set of codes might be selected for decoding complexity or
other practical reasons.  As an example, consider the collection $\mathcal{C}$ of $9191$ binary minimum distance codes of lengths 
between $6$ and $16$ given in \cite[DVD supplement]{KaskiO2006}.  We run the optimization
over the example in Sec.~\ref{sec:example} but restricting to $\mathcal{C}$.

The result obtained for epoch sizes is $(n_1 = 15, n_2 = 15, n_3 = 9, n_4 = 1)$.  Under the
Strassen approximation, this set of epoch sizes gives $5.344$ bits, as compared to $5.594$ bits under
the optimal epoch sizes under the Strassen approximation.  However the Strassen approximation is not
correct and the actual number of bits achieved with the optimized epoch sizes for $\mathcal{C}$ is
$7.246$ bits.  The two minimum distance codes used are the $(n = 15, M = 256, d = 5)$ code and 
the $(n = 9, M = 6, d = 3)$ code.  It remains to be seen whether the restriction to the collection
of minimum distance codes is actually suboptimal.

\section{Partial Ordering of Channels}
\label{sec:ordering}

It is of interest to order channels that die by quality.  The partial ordering
of DMCs was studied by Shannon \cite{Shannon1958}, and as a first step, we
can slightly extend his result to order channels that die having common death distributions.

\begin{defin}
Let $p(i,j)$ be the transition probabilities for a DMC $C_1$ and let $q(k,l)$ be the 
transition probabilities for a DMC $C_2$.  Then $C_1$ is said to include $C_2$, $C_1 \supseteq C_2$,
if there exist two sets of valid transition probabilities $r_{\gamma}(k,i)$ and $t_{\gamma}(j,l)$, and
there exists a vector $g$: $g_{\gamma} \ge 0$ and $\sum_{\gamma} g_{\gamma} = 1$, such that  
\[
\sum_{\gamma,i,j} g_{\gamma} r_{\gamma}(k,i) p(i,j) t_{\gamma}(j,l) = q(k,l) \mbox{.}
\]
\end{defin}
\begin{prop}
Consider two channels that die with identical death distributions:
$(\mathcal{X}_1,p_a,p_d,p_T(t),\mathcal{Y}_1)$ and $(\mathcal{X}_2,q_a,q_d,p_T(t),\mathcal{Y}_2)$.
Let DMC $C_1$ correspond to $p_a$ and let DMC $C_2$ correspond to $q_a$ and moreover suppose
that $C_1 \supseteq C_2$.  Fix a transmission time $N$ and an expected transmission volume $V$.
Let $\eta_1$ be the best level of reliability for the first channel and $\eta_2$ be the 
best level of reliability for the second channel, under $(N,V)$.
Then $\eta_1 \le \eta_2$.
\end{prop}
\begin{IEEEproof}
The main theorem of \cite{Shannon1958} proves that the average error probability
when transmitting an individual message code over $C_1$ is less than or equal to
the average error probability when transmitting the same individual message code over $C_2$.

Shannon's proof \cite{Shannon1958} holds \emph{mutatis mutandis} for maximum error probability, replacing ``average error
probability'' by ``maximum error probability.''

The desired result follows by concatenating individual message codes into a code.
\end{IEEEproof}

We can also order channels that die having common alive state transition probabilities.
\begin{defin}
Consider two random variables $T$ and $U$ with survival functions $R_T(\cdot)$ and $R_U(\cdot)$ respectively.
Then $U$ is said to stochastically dominate $T$, $U \ge_{{\rm st}} T$, if $R_T(t) \le R_U(t)$ for all $t$.
\end{defin}
\begin{prop}
Consider two channels that die with identical state properties:
$(\mathcal{X},p_a(y|x),p_d(y|x),p_T,\mathcal{Y})$ and $(\mathcal{X},p_a(y|x),p_d(y|x),q_U,\mathcal{Y})$.
Let death random variable $T$ correspond to $p_T$ and let death random variable $U$ correspond to $q_U$
and moreover suppose that $U \ge_{{\rm st}} T$.  Fix a transmission time $N$ and a level of
reliability $\eta$.  Let $V_1$ be the best expected transmission volume for the first channel and
$V_2$ be the best expected transmission volume for the second channel, under $(N,\eta)$.
Then $V_2 \ge V_1$.
\end{prop}
\begin{IEEEproof}
Recall the expected transmission volume expression \eqref{eq:optim} for the first channel:
\[
\max_{(n_i): \sum n_i = N}  \sum_{i} R_T(e_i) \log M^{*}(n_i,\eta)
\]
and for the second channel:
\[
\max_{(\nu_i): \sum \nu_i = N}  \sum_{i} R_U(\iota_i) \log M^{*}(\nu_i,\eta) \mbox{.}
\]
Since $R_T(t) \le R_U(t)$ for all $t$, the result follows directly.
\end{IEEEproof}

These two results give individual ordering principles in the two dimensions
essentially depicted in Fig.~\ref{fig:bscgeom}.  Putting them together provides a 
partial order on all channels that die: if one channel is better than another
channel in both dimensions, than it is better overall.
\begin{prop}
Consider two channels that die: $(\mathcal{X}_1,p_a,p_d,p_T,\mathcal{Y}_1)$ and $(\mathcal{X}_2,q_a,q_d,q_U,\mathcal{Y}_2)$.
Let DMC $C_1$ correspond to $p_a$ and let DMC $C_2$ correspond to $q_a$ and moreover suppose
that $C_2 \supseteq C_1$.   Let death random variable $T$ correspond to $p_T$ and let death random variable $U$ correspond to 
$q_U$ and moreover suppose that $U \ge_{{\rm st}} T$.  Fix a transmission time $N$ and a level of
reliability $\eta$.  Let $V_1$ be the best expected transmission volume for the first channel and
$V_2$ be the best expected transmission volume for the second channel, under $(N,\eta)$.
Then $V_2 \ge V_1$.
\end{prop}

\section{Conclusion and Future Work}
\label{sec:conc}
We have formulated the problem of communication
over channels that die and have shown how to maximize expected
transmission volume at a given level of error probability reliability.

There are several extensions to the basic formulation studied in this work
that one might consider; we list a few:
\begin{itemize}
  \item Inspired by synthetic biology \cite{CantonLE2008}, rather than thinking of
	death time as independent of the signaling scheme $X_1^n$, one might consider 
	channels that die because they lose fitness as a consequence of operation: $T$ would
	be dependent on $X_1^n$.  This would be similar to Gallager's panic button/child's toy channel,
	and would have intersymbol interference \cite{Gallager1968,Gallager1972}.  
	There would also be strong connections to channels that heat up \cite{KochLS2009} and 
	communication with a dynamic cost \cite[Ch.~3]{Eswaran2009}.

  \item In the emerging attention economy \cite{DavenportB2001}, agents faced with information overload 
	\cite{VanZandt2004} may permanently stop listening to certain communication media received over 
	noisy channels.  This setting is exactly modeled by channels that die.  The impact of communication
	over channels that die on the productivity and efficiency of human organizations may be determined
	by building on the results herein.

  \item Since channel death is indicated by the symbol $?$, the receiver unequivocally 
	knows death time.  Other channel models might not have a distinct output letter
	for death and would need to detect death, perhaps using the theory
	of estimating stopping times \cite{NiesenT2009}.

  \item Inspired by communication terminals that randomly lie within communication
	range, e.g. in vehicular communication, one might also consider a channel that is born at a 
	random time and then dies at a random time.  One would suspect that channel state feedback would be
	beneficial.  Networks of birth-death channels
	are also of interest and would have connections to percolation-style work \cite{Jacobs1959}.

  \item This work has simply considered the channel coding problem, however there are several
	formulations of end-to-end information transmission problems over channels that die, which
	are of interest in many application areas.  There is no reason to suspect a separation principle.
\end{itemize}
Randomly stepping back from infinity leads to some new understanding of the fundamental
limits of communication in the presence of noise and unreliability.

\section*{Acknowledgment}
We thank Barry Canton (Ginkgo BioWorks) and Drew Endy (Stanford University) 
for discussions on synthetic biology that initially inspired this work.  
Discussions with Baris Nakiboglu and Yury Polyanskiy (Princeton University)
are also appreciated.

\bibliographystyle{IEEEtran} 
\bibliography{abrv,conf_abrv,lrv_lib}

\begin{thebibliography}{10}
\providecommand{\url}[1]{#1}
\csname url@samestyle\endcsname
\providecommand{\newblock}{\relax}
\providecommand{\bibinfo}[2]{#2}
\providecommand{\BIBentrySTDinterwordspacing}{\spaceskip=0pt\relax}
\providecommand{\BIBentryALTinterwordstretchfactor}{4}
\providecommand{\BIBentryALTinterwordspacing}{\spaceskip=\fontdimen2\font plus
\BIBentryALTinterwordstretchfactor\fontdimen3\font minus
  \fontdimen4\font\relax}
\providecommand{\BIBforeignlanguage}[2]{{%
\expandafter\ifx\csname l@#1\endcsname\relax
\typeout{** WARNING: IEEEtran.bst: No hyphenation pattern has been}%
\typeout{** loaded for the language `#1'. Using the pattern for}%
\typeout{** the default language instead.}%
\else
\language=\csname l@#1\endcsname
\fi
#2}}
\providecommand{\BIBdecl}{\relax}
\BIBdecl

\bibitem{VarshneyMG2009}
L.~R. Varshney, S.~K. Mitter, and V.~K. Goyal, ``Channels that die,'' in
  \emph{Proc. 47th Annu. Allerton Conf. Commun. Control Comput.}, Sept.-Oct.
  2009, pp. 566--573.

\bibitem{Jacobs1959}
I.~M. Jacobs, ``Connectivity in probabilistic graphs: An abstract study of
  reliable communications in systems containing unreliable components,'' Sc.D.
  thesis, Massachusetts Institute of Technology, Cambridge, MA, Aug. 1959.

\bibitem{Davis1952}
D.~J. Davis, ``An analysis of some failure data,'' \emph{J. Am. Stat. Assoc.},
  vol.~47, no. 258, pp. 113--150, Jun. 1952.

\bibitem{DietrichD2009}
I.~Dietrich and F.~Dressler, ``On the lifetime of wireless sensor networks,''
  \emph{ACM Trans. Sensor Netw.}, vol.~5, no.~1, p.~5, Feb. 2009.

\bibitem{CantonLE2008}
B.~Canton, A.~Labno, and D.~Endy, ``Refinement and standardization of synthetic
  biological parts and devices,'' \emph{Nat. Biotechnol.}, vol.~26, no.~7, pp.
  787--793, Jul. 2008.

\bibitem{Pfeifer2006}
J.~Pfeifer, ``The use of information theory in biology: Lessons from social
  insects,'' \emph{Biol. Theory}, vol.~1, no.~3, pp. 317--330, Summer 2006.

\bibitem{Varshney2006}
L.~R. Varshney, ``Optimal information storage: Nonsequential sources and neural
  channels,'' S.M. thesis, Massachusetts Institute of Technology, Cambridge,
  MA, Jun. 2006.

\bibitem{Endy2005}
D.~Endy, ``Foundations for engineering biology,'' \emph{Nature}, vol. 438, no.
  7067, pp. 449--453, Nov. 2005.

\bibitem{Bekenstein2001}
J.~D. Bekenstein, ``The limits of information,'' \emph{Stud. Hist. Philos. Mod.
  Phys.}, vol.~32, no.~4, pp. 511--524, Dec. 2001.

\bibitem{Headrick1991}
D.~R. Headrick, \emph{The Invisible Weapon: Telecommunications and
  International Politics, 1851--1945}.\hskip 1em plus 0.5em minus 0.4em\relax
  New York: Oxford University Press, 1991.

\bibitem{HaraOAOM1996}
S.~Hara, A.~Ogino, M.~Araki, M.~Okada, and N.~Morinaga, ``Throughput
  performance of {SAW--ARQ} protocol with adaptive packet length in mobile
  packet data transmission,'' \emph{{IEEE} Trans. Veh. Technol.}, vol.~45,
  no.~3, pp. 561--569, Aug. 1996.

\bibitem{Modiano1999}
E.~Modiano, ``An adaptive algorithm for optimizing the packet size used in
  wireless {ARQ} protocols,'' \emph{Wireless Netw.}, vol.~5, no.~4, pp.
  279--286, Jul. 1999.

\bibitem{LettieriS1998}
P.~Lettieri and M.~B. Srivastava, ``Adaptive frame length control for improving
  wireless link throughput, range, and energy efficiency,'' in \emph{Proc. 17th
  Annu. Joint Conf. IEEE Computer Commun. Soc. (INFOCOM'98)}, vol.~2, Mar.
  1998, pp. 564--571.

\bibitem{CiSN2005}
S.~Ci, H.~Sharif, and K.~Nuli, ``Study of an adaptive frame size predictor to
  enhance energy conservation in wireless sensor networks,'' \emph{{IEEE} J.
  Sel. Areas Commun.}, vol.~23, no.~2, pp. 283--292, Feb. 2005.

\bibitem{Slepian1963}
D.~Slepian, ``Bounds on communication,'' \emph{Bell Syst. Tech. J.}, vol.~42,
  pp. 681--707, May 1963.

\bibitem{ChangHM1962}
S.~S.~L. Chang, B.~Harris, and J.~J. Metzner, ``Optimum message transmission in
  a finite time,'' \emph{{IRE} Trans. Inf. Theory}, vol.~8, no.~5, pp.
  215--224, Sep. 1962.

\bibitem{MacMullenC1998}
S.~J. MacMullan and O.~M. Collins, ``A comparison of known codes, random codes,
  and the best codes,'' \emph{{IEEE} Trans. Inf. Theory}, vol.~44, no.~7, pp.
  3009--3022, Nov. 1998.

\bibitem{Laneman2006}
J.~N. Laneman, ``On the distribution of mutual information,'' in \emph{Proc.
  Inf. Theory Appl. Inaugural Workshop}, Feb. 2006.

\bibitem{PolyanskiyPV2008}
Y.~Polyanskiy, H.~V. Poor, and S.~{Verd\'{u}}, ``New channel coding
  achievability bounds,'' in \emph{Proc. 2008 IEEE Int. Symp. Inf. Theory},
  Jul. 2008, pp. 1763--1767.

\bibitem{BuckinghamV2008}
D.~Buckingham and M.~C. Valenti, ``The information-outage probability of
  finite-length codes over {AWGN} channels,'' in \emph{Proc. 42nd Annu. Conf.
  Inf. Sci. Syst. (CISS 2008)}, Mar. 2008, pp. 390--395.

\bibitem{WiechmanS2008}
G.~Wiechman and I.~Sason, ``An improved sphere-packing bound for finite-length
  codes over symmetric memoryless channels,'' \emph{{IEEE} Trans. Inf. Theory},
  vol.~54, no.~5, pp. 1962--1990, May 2008.

\bibitem{OzarowSW1994}
L.~H. Ozarow, S.~Shamai, and A.~D. Wyner, ``Information theoretic
  considerations for cellular mobile radio,'' \emph{{IEEE} Trans. Veh.
  Technol.}, vol.~43, no.~2, pp. 359--378, May 1994.

\bibitem{Goldsmith2005}
A.~Goldsmith, \emph{Wireless Communications}.\hskip 1em plus 0.5em minus
  0.4em\relax New York: Cambridge University Press, 2005.

\bibitem{WolfWZ1970}
J.~K. Wolf, A.~D. Wyner, and J.~Ziv, ``The channel capacity of the postal
  channel,'' \emph{Inf. Control}, vol.~16, no.~2, pp. 167--172, Apr. 1970.

\bibitem{ZengZC2008}
M.~Zeng, R.~Zhang, and S.~Cui, ``On the outage capacity of a dying channel,''
  in \emph{Proc. IEEE Global Telecommun. Conf. (GLOBECOM 2008)}, Dec. 2008.

\bibitem{PolyanskiyPV2009}
Y.~Polyanskiy, H.~V. Poor, and S.~{Verd\'{u}}, ``Dispersion of the
  {G}ilbert-{E}lliott channel,'' in \emph{Proc. 2009 IEEE Int. Symp. Inf.
  Theory}, Jul. 2009, pp. 2209--2213.

\bibitem{PolyanskiyPV2009b}
------, ``Dispersion of {G}aussian channels,'' in \emph{Proc. 2009 IEEE Int.
  Symp. Inf. Theory}, Jul. 2009, pp. 2204--2208.

\bibitem{Shannon1956}
C.~E. Shannon, ``The zero error capacity of a noisy channel,'' \emph{{IRE}
  Trans. Inf. Theory}, vol. IT-2, no.~3, pp. 8--19, Sep. 1956.

\bibitem{Sanghavi2007}
S.~Sanghavi, ``Intermediate performance of rateless codes,'' in \emph{Proc.
  IEEE Inf. Theory Workshop (ITW'07)}, Sep. 2007, pp. 478--482.

\bibitem{KamraMFR2006}
A.~Kamra, V.~Misra, J.~Feldman, and D.~Rubenstein, ``Growth codes: Maximizing
  sensor network data persistence,'' in \emph{Proc. 2006 Conf. Appl. Technol.
  Archit. Protocols Comput. Commun. (SIGCOMM'06)}, Sep. 2006, pp. 255--266.

\bibitem{Gallager1968}
R.~G. Gallager, \emph{Information Theory and Reliable Communication}.\hskip 1em
  plus 0.5em minus 0.4em\relax New York: John Wiley \& Sons, 1968.

\bibitem{Ross1996}
S.~M. Ross, \emph{Stochastic Processes}.\hskip 1em plus 0.5em minus 0.4em\relax
  John Wiley \& Sons, 1996.

\bibitem{Gallager1996}
R.~G. Gallager, \emph{Discrete Stochastic Processes}.\hskip 1em plus 0.5em
  minus 0.4em\relax Boston: Kluwer Academic Publishers, 1996.

\bibitem{Gallager1972}
R.~Gallager, \emph{Information Theory and Reliable Communication}, ser.
  International Centre for Mechanical Sciences, Courses and Lectures.\hskip 1em
  plus 0.5em minus 0.4em\relax Vienna: Springer-Verlag, 1972, no.~30.

\bibitem{TatikondaM2009}
S.~Tatikonda and S.~Mitter, ``The capacity of channels with feedback,''
  \emph{{IEEE} Trans. Inf. Theory}, vol.~55, no.~1, pp. 323--349, Jan. 2009.

\bibitem{MushkinB1989}
M.~Mushkin and I.~Bar-David, ``Capacity and coding for the {G}ilbert--{E}lliott
  channels,'' \emph{{IEEE} Trans. Inf. Theory}, vol.~35, no.~6, pp. 1211--1290,
  Nov. 1989.

\bibitem{GoldsmithV1996}
A.~J. Goldsmith and P.~P. Varaiya, ``Capacity, mutual information, and coding
  for finite-state {M}arkov channels,'' \emph{{IEEE} Trans. Inf. Theory},
  vol.~42, no.~3, pp. 868--886, May 1996.

\bibitem{BratenT2002}
L.~E. Braten and T.~Tjelta, ``Semi-{M}arkov multistate modeling of the land
  mobile propagation channel for geostationary satellites,'' \emph{{IEEE}
  Trans. Antennas Propag.}, vol.~50, no.~12, pp. 1795--1802, Dec. 2002.

\bibitem{WangCA2008}
J.~Wang, J.~Cai, and A.~S. Alfa, ``New channel model for wireless
  communications: Finite-state phase-type semi-{M}arkov channel model,'' in
  \emph{Proc. IEEE Int. Conf. Commun. (ICC 2008)}, May 2008, pp. 4461--4465.

\bibitem{WangP2010}
S.~Wang and J.-T. Park, ``Modeling and analysis of multi-type failures in
  wireless body area networks with semi-{M}arkov model,'' \emph{{IEEE} Commun.
  Lett.}, vol.~14, no.~1, pp. 6--8, Jan. 2010.

\bibitem{Jelinek1968b}
F.~Jelinek, \emph{Probabilistic Information Theory: Discrete and Memoryless
  Models}.\hskip 1em plus 0.5em minus 0.4em\relax New York: McGraw-Hill Book
  Company, 1968.

\bibitem{Forney1974a}
G.~D. Forney, Jr., ``Convolutional codes {II}. {M}aximum-likelihood decoding,''
  \emph{Inf. Control}, vol.~25, no.~3, pp. 222--266, Jul. 1974.

\bibitem{McEliece2002}
R.~J. McEliece, \emph{The Theory of Information and Coding}.\hskip 1em plus
  0.5em minus 0.4em\relax Cambridge: Cambridge University Press, 2002.

\bibitem{KaskiO2006}
P.~Kaski and P.~R.~J. {\"{O}sterg{\aa}rd}, \emph{Classification Algorithms for
  Codes and Designs}.\hskip 1em plus 0.5em minus 0.4em\relax Berlin: Springer,
  2006.

\bibitem{BargM2005}
A.~Barg and A.~McGregor, ``Distance distribution of binary codes and the error
  probability of decoding,'' \emph{{IEEE} Trans. Inf. Theory}, vol.~51, no.~12,
  pp. 4237--4246, Dec. 2005.

\bibitem{Strassen1962}
V.~Strassen, ``Asymptotische absch{\"{a}}tzungen in {S}hannons
  informationstheorie,'' in \emph{Transactions of the 3rd Prague Conference on
  Information Theory, Statistical Decision Functions, Random Processes}.\hskip
  1em plus 0.5em minus 0.4em\relax Prague: Pub. House of the Czechoslovak
  Academy of Sciences, 1962, pp. 689--723.

\bibitem{Knuth1976}
D.~E. Knuth, ``Big omicron and big omega and big theta,'' \emph{SIGACT News},
  vol.~8, no.~2, pp. 18--24, Apr.-June 1976.

\bibitem{Weiss1960}
L.~Weiss, ``On the strong converse of the coding theorem for symmetric channels
  without memory,'' \emph{Q. Appl. Math.}, vol.~18, no.~3, pp. 209--214, Oct.
  1960.

\bibitem{Forney1968}
G.~D. Forney, Jr., ``Exponential error bounds for erasure, list, and decision
  feedback schemes,'' \emph{{IEEE} Trans. Inf. Theory}, vol. IT-14, no.~2, pp.
  206--220, Mar. 1968.

\bibitem{Polyanskiy2010}
Y.~Polyanskiy, ``Channel coding: non-asymptotic fundamental limits,'' Ph.D.
  dissertation, Princeton University, Nov. 2010.

\bibitem{Varshney2010}
L.~R. Varshney, ``Unreliable and resource-constrained decoding,'' Ph.D.~thesis,
  Massachusetts Institute of Technology, Cambridge, MA, Jun. 2010.

\bibitem{KhanKA1989}
M.~S.~A. Khan, A.~Khalique, and A.~M. Abouammoh, ``On estimating parameters in
  a discrete {W}eibull distribution,'' \emph{{IEEE} Trans. Rel.}, vol.~38,
  no.~3, pp. 348--350, Aug. 1989.

\bibitem{Bertsekas2005}
D.~P. Bertsekas, \emph{Dynamic Programming and Optimal Control}, 3rd~ed.\hskip
  1em plus 0.5em minus 0.4em\relax Belmont, MA: Athena Scientific, 2005,
  vol.~1.

\bibitem{Golay1949}
M.~J.~E. Golay, ``Notes on digital coding,'' \emph{Proc. {IRE}}, vol.~37,
  no.~6, p. 657, Jun. 1949.

\bibitem{Blahut1983}
R.~E. Blahut, \emph{Theory and Practice of Error Control Codes}.\hskip 1em plus
  0.5em minus 0.4em\relax Reading, MA: Addison-Wesley Publishing Company, 1983.

\bibitem{Shannon1958}
C.~E. Shannon, ``A note on a partial ordering for communication channels,''
  \emph{Inf. Control}, vol.~1, no.~4, pp. 390--397, Dec. 1958.

\bibitem{KochLS2009}
T.~Koch, A.~Lapidoth, and P.~P. Sotiriadis, ``Channels that heat up,''
  \emph{{IEEE} Trans. Inf. Theory}, vol.~55, no.~8, pp. 3594--3612, Aug. 2009.

\bibitem{Eswaran2009}
K.~Eswaran, ``Communication and third parties: Costs, cues, and
  confidentiality,'' Ph.D. dissertation, University of California, Berkeley,
  Berkeley, CA, 2009.

\bibitem{DavenportB2001}
T.~H. Davenport and J.~C. Beck, \emph{The Attention Economy: Understanding the
  New Currency of Business}.\hskip 1em plus 0.5em minus 0.4em\relax Boston:
  Harvard Business School Press, 2001.

\bibitem{VanZandt2004}
T.~Van~Zandt, ``Information overload in a network of targeted communication,''
  \emph{Rand J. Econ.}, vol.~35, no.~3, pp. 542--560, Autumn 2004.

\bibitem{NiesenT2009}
U.~Niesen and A.~Tchamkerten, ``Tracking stopping times through noisy
  observations,'' \emph{{IEEE} Trans. Inf. Theory}, vol.~55, no.~1, pp.
  422--432, Jan. 2009.

\end{thebibliography}

\end{document}